\documentclass[journal,10pt,twocolumn]{IEEEtran}
\usepackage{cite}
\usepackage{graphicx}  
\usepackage{tabulary}
\usepackage{amsfonts} 
\usepackage{amsmath}
\usepackage{chngpage}
\usepackage{makecell}
\usepackage{array}
\usepackage{amsthm}
\usepackage{booktabs}
\usepackage{amssymb}
\usepackage{multirow}
\usepackage{enumerate}
\usepackage{color}
\usepackage[dvipsnames]{xcolor}
\usepackage{mathrsfs}

\usepackage{float}
\usepackage{algorithm}  
\usepackage{algpseudocode}  
\usepackage{amsmath}  
\usepackage{url}

\usepackage{graphicx}
\usepackage{subfigure}
\usepackage[font=small]{caption}

\usepackage{mathtools}
\usepackage{enumitem}

\setlist[enumerate,1]{label=(\roman*)}

\newtheoremstyle{noparens}%
{}{}%
{\itshape}{}%
{\bfseries}{.}%
{ }%
{\thmname{#1}\thmnumber{ #2}\mdseries\thmnote{ #3}}

\theoremstyle{noparens}
\newtheorem{theorem}{Theorem}

\newtheorem{symmetry assumption}{Symmetry Assumption}
\newtheorem{lemma}{Lemma}

\newtheorem{definition}{Definition}

\begin{document}


\title{On the Second-Order Achievabilities of Indirect Quadratic Lossy Source Coding}

\author{
	\IEEEauthorblockN{Huiyuan Yang,~\IEEEmembership{Graduate Student Member,~IEEE}, Xiaojun Yuan,~\IEEEmembership{Senior Member,~IEEE}}
	\thanks{Huiyuan Yang and Xiaojun Yuan are with the National Key Laboratory on Wireless Communications, University of Electronic Science and Technology of China, Chengdu 611731, China (e-mail: hyyang@std.uestc.edu.cn; xjyuan@uestc.edu.cn).}
}

\maketitle
\IEEEpeerreviewmaketitle

\begin{abstract}
This paper studies the second-order achievabilities of indirect quadratic lossy source coding for a specific class of source models, where the term ``quadratic'' denotes that the reconstruction fidelity of the hidden source is quantified by a squared error distortion measure. Specifically, it is assumed that the hidden source $S$ can be expressed as $S = \varphi(X) + W$, where $X$ is the observable source with alphabet $\mathcal{X}$, $\varphi(\cdot)$ is a deterministic function, and $W$ is a random variable independent of $X$, satisfying $\mathbb{E}[W] = 0$, $\mathbb{E}[W^2] > 0$, $\mathbb{E}[W^3] = 0$, and $\mathbb{E}[W^6] < \infty$. Additionally, both the set $\{\varphi(x):\ x \in \mathcal{X} \}$ and the reconstruction alphabet for $S$ are assumed to be bounded. Under the above settings, a second-order achievability bound is established using techniques based on distortion-tilted information. This result is then generalized to the case of indirect quadratic lossy source coding with observed source reconstruction, where reconstruction is required for both the hidden source $S$ and the observable source $X$, and the distortion measure for $X$ is not necessarily quadratic. These obtained bounds are consistent in form with their finite-alphabet counterparts, which have been proven to be second-order tight.
\end{abstract}


\section{Introduction}

In this paper, we study the lossy compression of correlated memoryless sources $(S,X)$, where only source $X$ is observed by the encoder. We consider two cases of interest: After receiving the codeword through a rate-constrained noiseless channel, the decoder aims to
\begin{enumerate}
	\item \label{case_i}
	reconstruct only the hidden source $S$, or
	\item \label{case_ii}
	reconstruct both the sources $(S, X)$,
\end{enumerate}
with the lowest possible distortions under given distortion measures. Note that case \ref{case_i} corresponds to the well-known indirect (or noisy, remote, hidden) lossy source coding problem formulated in \cite{Dobrushin1962Information} and developed, e.g., by \cite{Wolf1970Transmission}, \cite{Witsenhausen1980Indirect}, and \cite{Kostina2016Nonasymptotic}. Case \ref{case_ii} is formulated in \cite{Liu2021Rate} for some semantic communication scenarios where observable data reconstruction and intrinsic state inference are simultaneously required. Both scenarios arise when only a noisy version of the target is observable, while in case \ref{case_ii}, the observable data is also of interest. The noise can be introduced by processes such as quantization, measurement, uncoded transmission, and some inherent random transition processes.

In this paper, we consider a specific class of source models and a squared error distortion measure for the hidden source $S$, as specified at the beginning of Section \ref{Preliminaries}. The focus is on investigating the second-order achievability for the aforementioned problems. While the second-order rates of these problems have been established in \cite{Kostina2016Nonasymptotic} and \cite{Yang2024Indirect} under the finite-alphabet setting, the type-based methods employed there are inherently inapplicable to our setup, where the alphabets are allowed to be continuous. Instead, we derive second-order achievable bounds by leveraging some matching properties of the considered source model and the squared error distortion measure, and by extending the general results for abstract-alphabet sources from \cite{Kostina2012Fixed} to our indirect settings. It turns out that our bounds take the same form as the second-order rates for finite-alphabet sources obtained in \cite{Kostina2016Nonasymptotic} and \cite{Yang2024Indirect}.


\subsection{Related Works}
\label{Related_Works}

The indirect lossy source coding problem was originally formulated in \cite{Dobrushin1962Information} in 1962, where it was shown that when the objective is to minimize average distortion, the indirect source coding problem with distortion measure $\textsf{d}_s$ is asymptotically equivalent to a surrogate direct source coding problem with distortion measure $\bar{\textsf{d}}_s$, given by 
\begin{align}
	\bar{\textsf{d}}_s(x,z) = \mathbb{E}[\textsf{d}_s(S,z)|X=x]. 
\end{align}
\cite{Witsenhausen1980Indirect} further demonstrated that the surrogate distortion measure argument implied in \cite{Dobrushin1962Information} is capable of addressing a broader class of lossy source coding problems. The surrogate distortion measure was also employed in \cite{Ephraim1988A} to extend the quantization theory for noise-free sources to that for noisy sources. The quantization of noisy sources has been further explored in works such as \cite{Ayanoglu1990On, Fischer1990Estimation, Linder1997Empirical}. Under the mean-squared error distortion, \cite{Sakrison1968Source} demonstrated that an asymptotically optimal coding procedure can be constructed by first estimating the hidden source sequence from the observed data and then encoding this estimate as if it were noise-free. Subsequently, \cite{Wolf1970Transmission} showed that the result in \cite{Sakrison1968Source} holds even in non-asymptotic settings. Additionally, as a special case of indirect lossy source coding under the logarithmic loss distortion measure, the information bottleneck problem was introduced in \cite{Tishby1999The} and has been extensively investigated in the literature (e.g., \cite{Chechik2005Information, Alemi2016Deep, Saxe2018On}) due to its profound connections with machine learning \cite{Goldfeld2020The}. 

The paradigm of simultaneously reconstructing observable and hidden sources was introduced recently in 2021 in \cite{Liu2021Rate, Liu2022Indirect}, where the hidden source $S$ was considered as the embedded semantics within the observable data $X$. This simultaneous reconstruction formulation was motivated by scenarios where both the observable data and its inferences (semantics) are required. Some examples include autonomous driving with raw data (e.g., videos) and its inferences (e.g., features) for human vision and machine processing, respectively, as well as speech signal compression for both lossy reconstruction and text transcription. In \cite{Liu2021Rate, Liu2022Indirect}, the corresponding rate-distortion function was derived by the unified treatment in \cite{Witsenhausen1980Indirect}, and some specializations were provided for cases such as quadratic Gaussian. The semantic modeling approach in \cite{Liu2021Rate, Liu2022Indirect} was subsequently applied to the multi-terminal source coding and joint source-channel coding scenarios in \cite{Yuxuan2023Rate} and \cite{Shi2023Excess}, respectively. \cite{Stavrou2023Role} and \cite{Li2023Fundamental} explored several analytical properties of the so-called semantic rate-distortion function. \cite{Stavrou2023Role} also gave a Blahut–Arimoto-type algorithm for computing the rate-distortion function, while \cite{Li2023Fundamental} further introduced a neural network for estimating the rate-distortion function from samples.


Research on finite blocklength and second-order analysis in source coding has its roots in the pioneering work \cite{Volker1962Asymptotische} of Strassen in 1962 in the context of lossless source coding (see also \cite{Kontoyiannis1997Second}). However, it was not until 2000 that the second-order analysis of lossy source coding was first investigated in \cite{Kontoyiannis2000Pointwise} (see also \cite{Kontoyiannis2002Arbitrary}), which considered variable-rate coding with the distortion required to remain below a given threshold. It is worth mentioning that the dispersion of lossy source coding, which captures the non-trivial component of the second-order term of the rate, was first introduced in \cite{Kontoyiannis2000Pointwise} as ``minimal coding variance.'' In the early 2010s, the second-order analysis of lossy source coding was revisited in \cite{Ingber2011Dispersion} and \cite{Kostina2012Fixed} under the excess-distortion probability constraint. \cite{Ingber2011Dispersion} focused on the finite-alphabet and quadratic Gaussian cases, while \cite{Kostina2012Fixed} established non-asymptotic achievability and converse bounds and, through asymptotic analysis of these bounds, obtained second-order rates for sources with abstract alphabets. After \cite{Ingber2011Dispersion} and \cite{Kostina2012Fixed}, the second-order asymptotics of various lossy source coding scenarios have been studied, such as in \cite{Kostina2013LossyJoint, Kostina2016Nonasymptotic, No2016Strong, Zhou2017Second, Zhou2019Refined, Zhou2019NonAsymptotic, Yang2024Indirect}. Thorough overviews can be found in \cite{Vincent2014Asymptotic} and \cite{Zhou2023Finite}. 

In particular, for indirect lossy source coding of stationary memoryless finite-alphabet sources with separable distortion measure, \cite{Kostina2016Nonasymptotic} showed that the minimum achievable codebook size $M$ compatible with blocklength $k$, maximum admissible distortion $d_s$, and excess distortion probability $\epsilon$ satisfies
\begin{align}
\label{second_order_general_introduction}
\log M^\star(k, d_s, \epsilon) =& k R_{S,X}(d_s) 
+ \sqrt{k \tilde{\mathcal{V}}(d_s)} Q^{-1}(\epsilon) + O(\log k), 
\end{align}
where $R_{S,X}(d_s)$ and $\tilde{\mathcal{V}}(d_s)$ are the rate-distortion and the rate-dispersion functions defined in \cite{Kostina2016Nonasymptotic}, and $Q^{-1}(\cdot)$ denotes the inverse of the complementary standard Gaussian cumulative distribution function. Similarly, the corresponding second-order result for indirect lossy source coding with observed source reconstruction is derived for stationary memoryless finite-alphabet sources in \cite{Yang2024Indirect}, and is given by
\begin{align}
\label{second_order_general_two_introduction}
\log M^\star(k, d_s, d_x, \epsilon) =& k R_{S,X}(d_s, d_x) 
+ \sqrt{k \tilde{\mathcal{V}}(d_s, d_x)} Q^{-1}(\epsilon)\nonumber \\ &+ O(\log k),
\end{align}
where the quantities in \eqref{second_order_general_two_introduction} are the counterparts of that in \eqref{second_order_general_introduction} and is defined in \cite{Yang2024Indirect}. However, as previously mentioned, for our source model that allows continuous alphabets, the type-based methods used in \cite{Kostina2016Nonasymptotic} and \cite{Yang2024Indirect} are no longer applicable. As will be demonstrated, we avoided using type-based methods and obtained achievability results that are consistent in form with \eqref{second_order_general_introduction} and \eqref{second_order_general_two_introduction}. We will discuss our proof and its relationship to the proofs in \cite{Kostina2016Nonasymptotic} and \cite{Yang2024Indirect} in Section \ref{Discussion}.

The remainder of this paper proceeds as follows. In Section \ref{Preliminaries}, we introduce basic definitions and properties. Section \ref{Second_Order_Asymptotics} lists our main results. In Section \ref{proof_theorem_Gaussian_approximation_achi}, we provide the proof of Theorem \ref{theorem_Gaussian_approximation_achi} and the corresponding discussion. Section \ref{Conclusion} concludes the paper. The proofs of Theorem \ref{theorem_Gaussian_approximation_achi_two} and several auxiliary lemmas are provided in the Appendices.

\section{Preliminaries}
\label{Preliminaries}

We consider correlated memoryless sources $(S, X)\sim P_{SX}=P_{X}P_{S|X}$, where we are given the distribution $P_X$ on alphabet $\mathcal{X}$ and the conditional distribution $P_{S|X}: \mathcal{X} \to \mathcal{S}$. Denote by $\hat{\mathcal{S}}$ the reconstruction alphabet of source $S$. Since the encoder can only observe source $X$, we refer to $X$ as the observable source and $S$ as the hidden source. We assume that 
\begin{equation}
	S = \varphi(X) + W,
\end{equation}
where $\varphi(\cdot)$ is a deterministic function, and the noise $W$ is independent of $X$ and satisfies 
\begin{align}
	\mathbb{E}[W] = 0,\
	\mathbb{E}[W^2] > 0,\
	\mathbb{E}[W^3] = 0,\
	\mathbb{E}[W^6] < \infty.
\end{align}
We additionally assume that both sets $\{\varphi(x):\ x \in \mathcal{X} \}$ and $\hat{\mathcal{S}}$ are bounded.


Denote the length-$k$ i.i.d. block of $(S, X)$ as $(S^k, X^k)$. We consider a separable squared error distortion measure $\textsf{d}_s:\ \mathcal{S}^k \times \hat{\mathcal{S}}^k \mapsto [0, +\infty]$ for $S^k$, i.e.,
\begin{align}
\label{quadratic_distortion}
\textsf{d}_s(s^k, z^k) \triangleq \frac{1}{k}\sum_{i=1}^k (s_i - z_i)^2,
\end{align}
where $z^k$ denotes the reconstruction of $s^k$, and $s_i$ and $z_i$ denote the $i$-th elements of $s^k$ and $z^k$, respectively. With a slight abuse of notation, we set
\begin{align}
	\textsf{d}_s(s, z) = (s-z)^2.
\end{align}

In the remainder of this section, we first introduce key definitions and properties pertaining to the case of solely recovering the hidden source in Subsection \ref{Indirect_Quadratic_Lossy_Source_Coding}, and then generalize the contents in Subsection \ref{Indirect_Quadratic_Lossy_Source_Coding} to the case of simultaneously recovering both the sources in Subsection \ref{Indirect_Quadratic_Lossy_Source_Coding_with_Observed_Source_Reconstruction}.

\subsection{Indirect Quadratic Lossy Source Coding}
\label{Indirect_Quadratic_Lossy_Source_Coding}

\begin{figure*}
	\centering
	\subfigure[The $(k, M, d_s, \epsilon)$ code.]{
		\begin{minipage}{0.6\textwidth}
			\includegraphics[width=\textwidth]{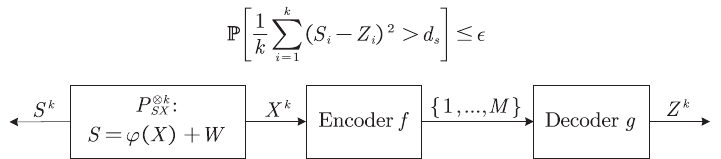} \\
	\end{minipage}
	\label{one_distortion_fig}}
	\subfigure[The $(k, M, d_s, d_x, \epsilon)$ code.]{
		\begin{minipage}{0.6\textwidth}
			\includegraphics[width=\textwidth]{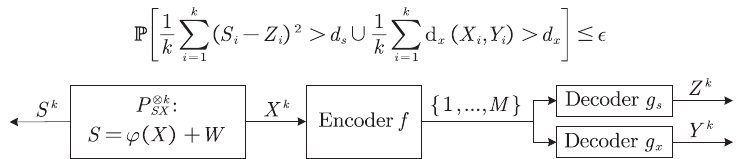} \\
	\end{minipage}
	\label{two_distortion_fig}}
	\caption{Indirect quadratic lossy source coding with/without observed source reconstruction.} 
	\label{JDSLC_figs}
\end{figure*}

\subsubsection{$(k, M, d_s, \epsilon)$ Code}
\label{The_Code_02}
As illustrated in Fig. \ref{one_distortion_fig}, the $(k, M, d_s, \epsilon)$ code is defined as follows.
\begin{definition}
	\label{definition_de_code}
	A $(k, M, d_s, \epsilon)$ code is a pair of mappings $\textsf{f}: \mathcal{X}^k \to \{1, \dots, M\}$, $\textsf{g}: \{1,\dots, M\} \to \hat{\mathcal{S}}^k$ such that $\mathbb{P}[\textsf{d}_s(S^k,\textsf{g}(\textsf{f}(X^k)))>d_s] \leq \epsilon$.
\end{definition}
The minimum achievable codebook size at blocklength $k$, maximum admissible distortion $d_s$, and excess distortion probability $\epsilon$, is defined as
\begin{equation}
M^\star(k, d_s, \epsilon) \triangleq \min \{M: \exists (k, M, d_s, \epsilon)\ \textrm{code}\}.
\end{equation}
In this paper, we aim to characterize the second-order achievabilities of $\log M^\star(k, d_s, \epsilon)$ as $k \to \infty$.

\subsubsection{Rate-Distortion Function}
The rate-distortion function of indirect lossy source coding can be defined by \cite{Dobrushin1962Information}
\begin{subequations}
	\label{rate_distortion_function_de}
	\begin{align}
	R_{S,X}(d_s) \triangleq \inf_{P_{Z|X}}\ &I(X;Z) \label{rate_distortion_function_obj_de} \\
	\mathrm{s.t.}\
	&\mathbb{E}\left[\bar{\textsf{d}}_s(X, Z)\right] \leq d_s, \label{Constraint_sur}
	\end{align}
\end{subequations}
where $\bar{\textsf{d}}_s:\ \mathcal{X} \times \hat{\mathcal{S}} \mapsto [0, +\infty]$ is given by
\begin{align}
	\bar{\textsf{d}}_s(x,z)\triangleq& \mathbb{E}[\textsf{d}_s(S, z)| X=x], \label{sura_d_s} \\
	=&(\varphi(x) - z)^2 + \sigma_w^2.
\end{align}
We also denote the corresponding separable distortion measure by
\begin{align}
\bar{\textsf{d}}_s(x^k,z^k)\triangleq \frac{1}{k}\sum_{i=1}^k \bar{\textsf{d}}_s(x_i,z_i). \label{sura_d_s_k}
\end{align}
Note that the constraint \eqref{Constraint_sur} implies that our considered indirect lossy source coding problem achieves the same performance with a surrogate (direct) lossy source coding problem with per-letter distortion measure $\bar{\textsf{d}}_s$ in the limit of infinite blocklengths \cite{Wolf1970Transmission,Witsenhausen1980Indirect}. We assume that $R_{S,X}(d_s) < \infty$ for some $d_s$ and define
\begin{align}
	d_{s, \min} \triangleq& \mathbb{E}[\inf_{z \in \widehat{\mathcal{S}}} \bar{\textsf{d}}_s(X,z)], \label{d_s_min}  \\
	d_{s, \max} \triangleq& \inf_{z \in \hat{\mathcal{S}}}\mathbb{E}[\bar{\textsf{d}}_s(X,z)], \label{d_s_max}
\end{align}
where the expectation in \eqref{d_s_max} is with respect to the unconditional distribution of $X$.


In the following sections, we show that the rate-distortion function $R_{S,X}(d_s)$ bounds the first-order term of $\log M^\star(k, d_s, \epsilon)/k$, i.e.,
\begin{equation}
\limsup\limits_{k\to \infty} \frac{\log M^\star(k, d_s, \epsilon)}{k} \leq R_{S,X}(d_s)
\end{equation}
under several mild conditions. 

\subsubsection{Distortion-Tilted Information}
For $d_s > d_{s, \min}$ and $d_s \neq d_{s, \max}$, the indirect distortion-tilted information $\tilde{\jmath}_{S,X}$ for the indirect lossy source coding problem, and the (direct) distortion-tilted information $\jmath_{X}$ for the surrogate (direct) lossy source coding problem, are respectively defined as \cite{Kostina2016Nonasymptotic, Kostina2012Fixed}
\begin{align}
\tilde{\jmath}_{S,X}(s,x,z,d_s) \triangleq& \imath_{X;Z^\star}(x;z) + \lambda_s^\star \textsf{d}_s(s,z) - \lambda_s^\star d_s, \label{tilted_information} \\
\jmath_{X}(x,d_s) 
\triangleq& \log \frac{1}{\mathbb{E}[\exp\{\lambda_s^\star d_s - \lambda_s^\star \bar{\textsf{d}}_s(x,Z^\star)\}]}, \label{direct_tilted_information}
\end{align}
where $P_{Z^\star| X}$ achieves the rate-distortion function $R_{S,X}(d_s)$, $P_{Z^\star}=\mathbb{E}_{X\sim P_X}[P_{Z^\star| X}]$,
\begin{align}
\imath_{X;Z^\star}(x;z) \triangleq& \log \frac{\mathrm{d} P_{Z^\star| X=x}}{\mathrm{d} P_{Z^\star}}(z), \label{information_density} \\
\lambda_s^\star \triangleq& - R_{S,X}^{\prime}(d_s), \label{lambda_s}
\end{align}
and the expectation in \eqref{direct_tilted_information} is with respect to the unconditional distribution $P_{Z^\star}$. By the definition \eqref{tilted_information}, we have 
\begin{align}
	\mathbb{E}[\tilde{\jmath}_{S,X}(S,X,Z^{\star},d_s)] = R_{S,X}(d_s).
\end{align}
It is worth noting that $\jmath_{X}(x,d_s)$ can be expressed in a form similar to $\tilde{\jmath}_{S,X}(s,x,z,d_s)$, i.e.,\cite[Theorem 2.1]{Kostina2013Lossy}, \cite[Property 1]{Kostina2012Fixed}
\begin{align}
\label{direct_tilted_information_as}
\jmath_{X}(x,d_s) 
=\imath_{X;Z^\star}(x;z) + \lambda_s^\star \bar{\textsf{d}}_s(x,z) - \lambda_s^\star d_s
\end{align}
for $P_{Z^\star}$-almost every $z$. As a consequence,
\begin{align}
\label{E_j_equals_R}
	\mathbb{E}[\jmath_{X}(X,d_s)] = R_{S,X}(d_s).
\end{align}
By \eqref{tilted_information} and \eqref{direct_tilted_information_as}, we have
\begin{equation}
	\jmath_{X}(x,d_s) = \mathbb{E}[\tilde{\jmath}_{S,X}(S,x,z,d_s)|X=x]
\end{equation}
for $P_{Z^\star}$-almost every $z$.

Define the indirect rate-dispersion function, $\tilde{\mathcal{V}}(d_s)$, and the direct rate-dispersion function, $\mathcal{V}(d_s)$, as
\begin{align}
\tilde{\mathcal{V}}(d_s) \triangleq& \mathrm{Var}\left[\tilde{\jmath}_{S,X}(S,X,Z^\star,d_s)\right], \label{indirect_rate_dispersion_func} \\
\mathcal{V}(d_s) \triangleq& \mathrm{Var}\left[\jmath_{X}(X,d_s)\right]. \label{direct_rate_dispersion_func}
\end{align}
Note that by the law of total variance,
\begin{equation}
\label{relationship_V_tildeV}
\tilde{\mathcal{V}}(d_s) = \mathcal{V}(d_s) + \lambda_s^{\star 2} \mathbb{E}\left[\mathrm{Var}\left[\textsf{d}_s(S,Z^\star)|X,Z^\star\right]\right],
\end{equation}
where $\mathrm{Var}\left[U|V\right] \triangleq \mathbb{E}\left[(U - \mathbb{E}\left[U|V\right])^2| V \right]$. Observe that the difference between $\tilde{\mathcal{V}}(d_s)$ and $\mathcal{V}(d_s)$ arises from the residual uncertainty about $S$ given $X$.

In the following sections, we will show that the indirect rate-dispersion function $\tilde{\mathcal{V}}(d_s)$ plays a crucial role in bounding the second-order term of $\log M^\star(k, d_s, \epsilon)/k$, i.e.,
\begin{equation}
\limsup\limits_{k\to \infty} \frac{\log M^\star(k, d_s, \epsilon) - kR_{S,X}(d_s)}{\sqrt{k}} \leq \sqrt{\tilde{\mathcal{V}}(d_s)}Q^{-1}(\epsilon)
\end{equation}
under several mild conditions.

\subsection{Indirect Quadratic Lossy Source Coding with Observed Source Reconstruction}
\label{Indirect_Quadratic_Lossy_Source_Coding_with_Observed_Source_Reconstruction}

In this subsection, we generalize the content in the previous subsection to the case of simultaneously recovering both the hidden and observable sources. 

Specifically, denote by $\hat{\mathcal{X}}$ the reconstruction alphabet of source $X$. We consider per-letter distortion measure $\textsf{d}_x: \mathcal{X} \times \hat{\mathcal{X}} \mapsto [0, +\infty)$ for the observable source $X$ and consider a separable distortion measure for block $X^k$, i.e., with a slight abuse of notation,
\begin{align}
\label{d_x_distortion_measure}
\textsf{d}_x(x^k, y^k) \triangleq \frac{1}{k}\sum_{i=1}^k \textsf{d}_x(x_i, y_i).
\end{align}
Note that, unlike $\textsf{d}_s$, we consider $\textsf{d}_x$ to be a general distortion measure, without the assumption of being quadratic. The corresponding code is defined as follows.
\begin{definition}
	\label{definition_general_code}
	A $(k, M, d_s, d_x, \epsilon)$ code is a triplet of mappings $\textsf{f}: \mathcal{X}^k \to \{1, \dots, M\}$, $\textsf{g}_s: \{1,\dots, M\} \to \hat{\mathcal{S}}^k$, and $\textsf{g}_x: \{1,\dots, M\} \to \hat{\mathcal{X}}^k$ such that the joint excess distortion probability $\mathbb{P}[\textsf{d}_s(S^k,\textsf{g}_s(\textsf{f}(X^k)))>d_s \cup \textsf{d}_x(X^k,\textsf{g}_x(\textsf{f}(X^k)))>d_x] \leq \epsilon$.
\end{definition}
The $(k, M, d_s, d_x, \epsilon)$ is illustrated in Fig. \ref{two_distortion_fig}. The minimum achievable codebook size with parameters $k$, $d_s$, $d_x$, $\epsilon$ is defined as
\begin{equation}
M^\star(k, d_s, d_x, \epsilon) \triangleq \min \{M: \exists (k, M, d_s, d_x, \epsilon)\ \textrm{code}\}.
\end{equation}
This paper seeks to provide a second-order achievability bound for $\log M^\star(k, d_s, d_x, \epsilon)$ as $k \to \infty$.

The rate-distortion function for the problem of indirect quadratic lossy source coding with observed source reconstruction is given by
\begin{subequations}
	\label{rate_distortion_function}
	\begin{align}
	R_{S,X}(d_s,d_x) \triangleq \inf_{P_{ZY|X}}\ &I(X;Z,Y) \label{rate_distortion_function_obj} \\
	\mathrm{s.t.}\
	&\mathbb{E}\left[\bar{\textsf{d}}_s(X, Z)\right] \leq d_s, \\
	&\mathbb{E}\left[\textsf{d}_x(X,Y)\right] \leq d_x. 
	\end{align}
\end{subequations}
Define $\mathcal{D}_{\mathrm{adm}} \triangleq \{(d_s,d_x):  R_{S,X}(d_s,d_x) < \infty \}$ and we assume that $D_{\mathrm{adm}}$ is non-empty. 
Denote by $\mathcal{P}^{\star}(d_s,d_x)$ the set of optimal solutions of problem \eqref{rate_distortion_function} associated with $(d_s,d_x) \in \mathcal{D}_{\mathrm{adm}}$.  Define set
\begin{align}
\mathcal{D}_{nd} \triangleq \{&(d_s,d_x) \in \mathcal{D}_{\mathrm{adm}}: \forall P_{Z^\star Y^\star|X} \in \mathcal{P}^{\star}(d_s,d_x),\nonumber \\ &\mathbb{E}\left[\bar{\textsf{d}}_s(X,Z^\star)\right] = d_s \textrm{ and }  \mathbb{E}\left[\textsf{d}_x(X,Y^\star)\right] = d_x\}. \label{Definition_of_Dnd}
\end{align}
Note that for all $(d_s,d_x) \notin \mathcal{D}_{nd}$, at least one of the constraints in problem \eqref{rate_distortion_function} would be loose, leading to a degeneration of the optimization problem. This work focuses exclusively on the non-degenerate case, specifically on $(d_s, d_x) \in \textsf{int}(\mathcal{D}_{nd})$, where $\textsf{int}(\cdot)$ denotes the interior of the input set.

For $(d_s, d_x) \in \textsf{int}(\mathcal{D}_{nd})$, the distortion-tilted informations are defined as follows:
\begin{equation}
\label{tilted_information_two}
\begin{aligned}
\tilde{\jmath}_{S,X}(s,x,z,y,d_s,d_x) \triangleq& \imath_{X;Z^\star Y^\star}(x;z,y) + \lambda_s^\star \textsf{d}_s(s,z)\\ 
&+ \lambda_x^\star \textsf{d}_x(x,y) - \lambda_s^\star d_s - \lambda_x^\star d_x,
\end{aligned}
\end{equation}
\begin{align}
\label{direct_tilted_information_two}
&\jmath_{X}(x,d_s,d_x) \nonumber \\ 
\triangleq &\log\! \frac{1}{\mathbb{E}[\exp\{\lambda_s^\star d_s \! + \! \lambda_x^\star d_x \!-\! \lambda_s^\star \bar{\textsf{d}}_s(x,Z^\star) \!-\! \lambda_x^\star \textsf{d}_x(x,Y^\star)\}]},
\end{align}
where $P_{Z^\star Y^\star| X}$ achieves the rate-distortion function $R_{S,X}(d_s,d_x)$, $P_{Z^\star Y^\star}=\mathbb{E}_{X\sim P_X}[P_{Z^\star Y^\star| X}]$,
\begin{equation}
\label{information_density_two}
\imath_{X;Z^\star Y^\star}(x;z,y) \triangleq \log \frac{\mathrm{d} P_{Z^\star Y^\star| X=x}}{\mathrm{d} P_{Z^\star Y^\star}}(z,y),
\end{equation}
\begin{equation}
\label{lambda_s_two}
\lambda_s^\star \triangleq -\frac{\partial R_{S,X}(d_s,d_x)}{\partial d_s},
\end{equation}
\begin{equation}
\label{lambda_x}
\lambda_x^\star \triangleq -\frac{\partial R_{S,X}(d_s,d_x)}{\partial d_x},
\end{equation}
and the expectation in \eqref{direct_tilted_information_two} is with respect to the unconditional distribution $P_{Z^\star Y^\star}$. Similar to \eqref{direct_tilted_information_as}, for $P_{Z^\star Y^\star}$-almost every $(z,y)$, it holds that \cite[Property 1]{Yang2024Indirect}
\begin{equation}
\label{d_tilted_information_of_surrogate}
\begin{aligned}
\jmath_{X}(x,d_s,d_x) = &\imath_{X;Z^\star Y^\star}(x;z,y) + \lambda_s^\star \bar{\textsf{d}}_s(x,z)\\ 
&+ \lambda_x^\star \textsf{d}_x(x,y) - \lambda_s^\star d_s - \lambda_x^\star d_x.
\end{aligned}
\end{equation}
Consequently, we have 
\begin{align}
\label{E_j_equals_R_two}
\mathbb{E}[\jmath_{X}(X,d_s, d_x)] = R_{S,X}(d_s, d_x).
\end{align}

Define the rate-dispersion functions, $\tilde{\mathcal{V}}(d_s,d_x)$ and $\mathcal{V}(d_s,d_x)$, as \cite{Yang2024Indirect}
\begin{align}
\tilde{\mathcal{V}}(d_s,d_x) \triangleq& \textrm{Var}\left[\tilde{\jmath}_{S,X}(S,X,Z^\star,Y^\star,d_s,d_x)\right], \label{indirect_rate_dispersion_func_two} \\
\mathcal{V}(d_s,d_x) \triangleq& \textrm{Var}\left[\jmath_{X}(X,d_s,d_x)\right]. \label{direct_rate_dispersion_func_two}
\end{align}
Note that \cite[Proposition 1]{Yang2024Indirect}
\begin{equation}
\label{relationship_V_tildeV_two}
\tilde{\mathcal{V}}(d_s,d_x) = \mathcal{V}(d_s,d_x) + \lambda_s^{\star 2}\mathbb{E}\left[\textrm{Var}\left[\textsf{d}_s(S,Z^\star)|X,Z^\star\right]\right].
\end{equation}

\section{Main Results: Second-Order Achievabilities}
\label{Second_Order_Asymptotics}

In this section, we present the second-order achievability results for the codes defined in Subsections \ref{Indirect_Quadratic_Lossy_Source_Coding} and \ref{Indirect_Quadratic_Lossy_Source_Coding_with_Observed_Source_Reconstruction}, respectively. For the case of solely recovering the hidden source, we have the following theorem. 
\begin{theorem}
	\label{theorem_Gaussian_approximation_achi}
	(Second-order Achievability for $(k, M, d_s, \epsilon)$ Codes): Fix $0 < \epsilon < 1$ and $d_{s, \min} < d_s < d_{s, \max}$. 
	Under the source model and distortion measure introduced in Section \ref{Preliminaries}, we have
	\begin{align}
	\label{second_order_general}
	\log M^\star(k, d_s, \epsilon) \leq& k R_{S,X}(d_s) 
	+ \sqrt{k \tilde{\mathcal{V}}(d_s)} Q^{-1}(\epsilon)\nonumber \\ &+ O(\log k),
	\end{align}
	where $f(k)=O(g(k))$ means $\lim\sup_{k \to \infty}\left|f(k)/g(k)\right| < \infty$.
\end{theorem}
\begin{proof}[Proof]
	See Section \ref{proof_theorem_Gaussian_approximation_achi}.
\end{proof}
The result for simultaneously recovering both the hidden and observable sources follows a form similar to that of recovering only the hidden source, as shown in Theorem \ref{theorem_Gaussian_approximation_achi_two}.
\begin{theorem}
	\label{theorem_Gaussian_approximation_achi_two}
	(Second-order Achievability for $(k, M, d_s, d_x, \epsilon)$ Codes): Fix $0 < \epsilon < 1$ and $(d_s, d_x) \in \textsf{int}(\mathcal{D}_{nd})$. Under the source model and distortion measure introduced in Section \ref{Preliminaries}, we have
	\begin{align}
	\label{second_order_general_two}
	\log M^\star(k, d_s, d_x, \epsilon) \leq& k R_{S,X}(d_s, d_x) 
	+ \sqrt{k \tilde{\mathcal{V}}(d_s, d_x)} Q^{-1}(\epsilon)\nonumber \\ &+ O(\log k).
	\end{align}
\end{theorem}
\begin{proof}[Proof]
	See Appendix \ref{proof_theorem_Gaussian_approximation_achi_two}.
\end{proof}
The bounds established in Theorem \ref{theorem_Gaussian_approximation_achi} and Theorem \ref{theorem_Gaussian_approximation_achi_two} for quadratic lossy source coding exhibit a form analogous to the second-order tight results (i.e., \eqref{second_order_general_introduction} and \eqref{second_order_general_two_introduction}) for finite-alphabet sources presented in \cite[Theorem 5]{Kostina2016Nonasymptotic} and \cite[Theorem 4]{Yang2024Indirect}, respectively.

\section{Proof of Theorem \ref{theorem_Gaussian_approximation_achi}} \label{proof_theorem_Gaussian_approximation_achi}

In this section, we first present the proof of Theorem \ref{theorem_Gaussian_approximation_achi} in Subsection \ref{Proof_of_Theorem_1}, followed by a discussion of the proof in Subsection \ref{Discussion}.

\subsection{Proof of Theorem \ref{theorem_Gaussian_approximation_achi}}
\label{Proof_of_Theorem_1}

We now begin our proof, which is an asymptotic analysis of Lemma \ref{Single_Shot_Achi_one_c} given in Appendix \ref{Auxiliary_Results}. By Lemma \ref{Single_Shot_Achi_one_c}, there exists a $(k, M, d_s, \epsilon')$ code such that
\begin{equation}
\label{achievability_1}
\epsilon' \leq \int_{0}^1 \mathbb{E}\Big[\mathbb{P}^M\left[\pi(X^k, Z^{\star k}) > t | X^k \right]\Big] \mathrm{d} t,
\end{equation}
where function $\pi$ is defined in \eqref{pi_func} in Lemma \ref{Single_Shot_Achi_one_c}, $P_{X^k Z^{\star k}} = P_{X^k} P_{Z^{\star k}}$, $P_{X^k} = P_X \times \dots \times P_X$, $P_{Z^{\star k}} = P_{Z^{\star}} \times \dots \times P_{Z^{\star}}$ with $P_{Z^{\star}}$ achieving the rate-distortion function $R_{S,X}(d_s)$.\footnote{$P_{Z^{\star}}$ achieves $R_{S,X}(d_s)$ means that $P_{Z^\star}=\mathbb{E}_{X\sim P_X}[P_{Z^\star| X}]$, where $P_{Z^\star| X}$ achieves $R_{S,X}(d_s)$.}
For any $\gamma > 0$, we have
\begin{align}
&\int_{0}^1 \mathbb{E}\Big[\mathbb{P}^M\left[\pi(X^k, Z^{\star k}) > t | X^k \right]\Big] \mathrm{d} t \nonumber \\
\leq& e^{-\frac{M}{\gamma}} + \int_{0}^1 \mathbb{P}\Big[\gamma \mathbb{P}\left[\pi(X^k, Z^{\star k}) \leq t | X^k \right] < 1\Big] \mathrm{d} t \label{achi_01} \\
\leq& e^{-\frac{M}{\gamma}} + \int_{\frac{2 \zeta}{\sqrt{k}}}^{1} \mathbb{P}\Big[\gamma \mathbb{P}\left[\pi(X^k, Z^{\star k}) \leq t | X^k \right] < 1\Big] \mathrm{d} t \nonumber \\ 
&+\frac{2 \zeta}{\sqrt{k}}, \label{achi_02}
\end{align}
where \eqref{achi_01} holds by applying $(1 -p)^M \leq e^{-Mp} \leq e^{-\frac{M}{\gamma}}\min\{1, \gamma p\} + 1\{\gamma p < 1\}$ \cite{Yury2012Information, Kostina2016Nonasymptotic}, and in \eqref{achi_02}, we set 
\begin{equation}
\label{zeta}
\zeta = \frac{c_0 T_0}{\sigma_{w^2}^3},
\end{equation}
where the constants $c_0$, $T_0$, $\sigma_{w^2} > 0$ will be specified later.
Let 
\begin{align}
\log M =& \log \gamma + \log \ln \sqrt{k}, \label{log_M} \\
\log \gamma =& kR_{S,X}(d_s) + \sqrt{k \tilde{\mathcal{V}}(d_s)} Q^{-1}(\epsilon) + O(\log k), \label{log_gamma}
\end{align}
with a properly chosen $O(\log k)$. In the following, we show that the right-hand side of \eqref{achi_02} is upper-bounded by $\epsilon$, which implies $\epsilon' \leq \epsilon$, thereby completing the proof.

Let
\begin{align}
\theta(x^k, z^k) \triangleq& \frac{1}{k}\sum_{i=1}^k(\varphi(x_i) - z_i)^2,
\end{align}
\begin{align}
\mu_k(x^k,z^k) \triangleq& \frac{1}{k}\sum_{i=1}^k \mathbb{E}\big[\textsf{d}_s(S_i, z_i)\big|X_i = x_i\big] \nonumber \\ 
=&\theta(x^k, z^k) + \sigma_w^2, \label{mu_k}
\end{align}
\begin{align}
V_k(x^k,z^k) \triangleq& \frac{1}{k}\sum_{i=1}^k \mathrm{Var}\big[\textsf{d}_s(S_i, z_i)\big|X_i = x_i\big] \nonumber \\ 
=&4 \sigma_w^2 \theta(x^k, z^k) + \sigma_{w^2}^2, \label{V_k} 
\end{align}
\begin{align}
&T_k(x^k, z^k) \nonumber \\
\triangleq& \frac{1}{k}\sum_{i=1}^k \mathbb{E}\big[|\textsf{d}_s(S_i, z_i) - \mathbb{E}[\textsf{d}_s(S_i, z_i)|X_i=x_i]|^3\big|X_i = x_i\big] \nonumber \\
=&\frac{1}{k}\sum_{i=1}^k \mathbb{E}\big[|W^2 + 2(\varphi(x_i) - z_i)W - \sigma_w^2|^3\big], \label{T_k}
\end{align}
\begin{align}
B_k(x^k,z^k) \triangleq& \frac{c_0 T_k(x^k, z^k)}{(V_k(x^k,z^k))^{3/2}}, \label{B_k}
\end{align}
where
\begin{align}
\sigma_w \triangleq \sqrt{\mathrm{Var}[W]},\
\sigma_{w^2} \triangleq \sqrt{\mathrm{Var}[W^2]},
\end{align}
and $c_0$ is that in the Berry–Ess\'een theorem.
By the assumption that $\mathbb{E}[W^6] < \infty$ and both sets $\{\varphi(x):\ x \in \mathcal{X} \}$ and $\hat{\mathcal{S}}$ are bounded, there exists a constant $T_0$ such that
\begin{equation}
	T_k(x^k, z^k) \leq T_0, \ \forall x^k \in \mathcal{X}^k,\ z^k \in \hat{\mathcal{S}}^k,\ k > 0.
\end{equation}
Juxtaposing this with the observation that $V_k(x^k,z^k) \geq \sigma_{w^2}^2$, we have
\begin{equation}
\label{Upper_Bounded_B}
B_k(x^k,z^k) \leq \zeta, \ \forall x^k \in \mathcal{X}^k,\ z^k \in \hat{\mathcal{S}}^k,\ k > 0.
\end{equation}

Given the above preparations, for $(x^k,z^k)$ and $t$ such that
\begin{align}
&\mu_k(x^k,z^k) \leq d_s - \sqrt{\frac{V_k(x^k,z^k)}{k}}Q^{-1}\Big(t - \frac{\zeta}{\sqrt{k}}\Big), \label{core_condition_02} \\
& \frac{\zeta}{\sqrt{k}} < t < 1, \label{core_condition_02_01}
\end{align}
we have
\begin{align}
&\pi(x^k,z^k) \nonumber \\ 
\leq&\mathbb{P}\Bigg[\textsf{d}_s(S^k, z^k) > \mu(x^k,z^k) \nonumber \\ 
&+ \sqrt{\frac{V(x^k,z^k)}{k}}Q^{-1}\Big(t - \frac{\zeta}{\sqrt{k}}\Big) \Bigg| X^k=x^k\Bigg] \label{By_Condition_1} \\
\leq&\mathbb{P}\Bigg[\textsf{d}_s(S^k, z^k) > \mu(x^k,z^k) \nonumber \\ 
&+ \sqrt{\frac{V(x^k,z^k)}{k}}Q^{-1}\Bigg(t - \frac{B_k(x^k,z^k)}{\sqrt{k}}\Bigg) \Bigg| X^k=x^k\Bigg] \label{By_max_min} \\
\leq&t, \label{By_berry_esseen}
\end{align}
where \eqref{By_Condition_1} is by \eqref{core_condition_02}, \eqref{By_max_min} is by \eqref{Upper_Bounded_B}, and \eqref{By_berry_esseen} is by the Berry–Ess\'een theorem. The sufficient conditions \eqref{core_condition_02} and \eqref{core_condition_02_01} for \eqref{By_berry_esseen} allows us to relax the bound in \eqref{achi_02}. However, prior to this, we can further strengthen condition \eqref{core_condition_02} to a more tractable form using the following lemma.
\begin{lemma} 
	\label{tackle_V_lemma}
	Fix $d_s > d_{s, \min}$. There exists a constant $k_0$ such that for all $k > k_0$ and $t$ satisfying $\frac{2\zeta}{\sqrt{k}} \leq t \leq 1$, we have \eqref{core_condition_02} holds if
	\begin{align}
	\theta(x^k, z^k) \leq d_s - \sigma_w^2 - \delta_s(t, k),
	\end{align}
	where
	\begin{align}
	\delta_s(t, k) \triangleq& \sqrt{\frac{V_1}{k}}Q^{-1}\Big(t - \frac{\zeta}{\sqrt{k}}\Big), \label{delta__s_t_k} \\
	V_1 \triangleq& 4\sigma_w^2 \big(d_s - \sigma_w^2\big) + \sigma_{w^2}^2. \label{V_1}
	\end{align}
\end{lemma}
\begin{proof}[Proof]
	See Appendix \ref{proof_tackle_V_lemma}.
\end{proof}
Note that when $\lambda_s^{\star} > 0$, by the complementary slackness condition, we have
\begin{align}
	\mathbb{E}\left[  (\varphi(X) - Z^{\star})^2 \right] + \sigma_w^2 = \mathbb{E}\left[\bar{\textsf{d}}_s(X, Z^{\star})\right] = d_s.
\end{align}
Juxtaposing this with \eqref{V_1} and
\begin{align}
	\mathbb{E}\left[\mathrm{Var}\left[\textsf{d}_s(S,Z^\star)|X,Z^\star\right]\right] = 4 \sigma_w^2 \mathbb{E}\left[  (\varphi(X) - Z^{\star})^2 \right] + \sigma_{w^2}^2,
\end{align}
when $d_{s, \min} < d_s < d_{s, \max}$, which implies that $\lambda_s^{\star} > 0$, we have
\begin{align}
\label{EVAR_e_V1}
	V_1 = \mathbb{E}\left[\mathrm{Var}\left[\textsf{d}_s(S,Z^\star)|X,Z^\star\right]\right].
\end{align}


Define
\begin{align}
&g_{Z^{\star k}}(x^k, t) \nonumber \\
\triangleq& \mathbb{P}\bigg[\frac{1}{k}\sum_{i=1}^k(\varphi(x_i) - Z^{\star}_i)^2 + \sigma_w^2 \leq d_s - \delta_s(t, k) \bigg],
\end{align}
where $Z^{\star k}$ is distributed according to the unconditional distribution $P_{Z^{\star k}}$. By \eqref{core_condition_02}-\eqref{By_berry_esseen} and Lemma \ref{tackle_V_lemma}, for all $k > k_0$ and $\frac{2\zeta}{\sqrt{k}} \leq t \leq 1$, we have
\begin{align}
\label{core_01}
\mathbb{P}[\pi(x^k, Z^{\star k}) \leq t] 
\geq g_{Z^{\star k}}(x^k, t),
\end{align}
where $k_0$ is that in Lemma \ref{tackle_V_lemma}. Consequently, for all $k > k_0$, we can relax the second term of \eqref{achi_02} by
\begin{align}
&\int_{\frac{2 \zeta}{\sqrt{k}}}^{1} \mathbb{P}\Big[\gamma \mathbb{P}\left[\pi(X^k, Z^{\star k}) \leq t | X^k \right] < 1\Big] \mathrm{d} t \nonumber \\
\leq&\int_{\frac{2 \zeta}{\sqrt{k}}}^{1} \mathbb{P}\Big[\gamma g_{Z^{\star k}}(X^k, t) < 1\Big] \mathrm{d} t. \label{pro_01}
\end{align}


We now give a lemma to establish a connection between $g_{Z^{\star k}}(x^k, t)$ and the distortion-tilted information $\jmath_{X}(x,d_s)$.

\begin{lemma} 
	\label{Lemma_g_d_tilted_info}
	There exist constants $C$, $T$, $k_6 > 0$ such that for all $k > k_6$, $d_{s, \min} < d_s < d_{s,\max}$, and $\frac{2\zeta}{\sqrt{k}} \leq t \leq 1$,
	\begin{align}
	\label{Fixed_result}
	\mathbb{P}\bigg[&\log\frac{1}{g_{Z^{\star k}}(X^k, t)} \leq \sum_{i=1}^k \jmath_{X}(X_i,d_s) + k \lambda_s^{\star}\delta_s(t, k) \nonumber \\  
	&+ C \log k \bigg] 
	\geq 1 - \frac{T}{\sqrt{k}}.
	\end{align}
\end{lemma}
\begin{proof}[Proof]
	See Appendix \ref{proof_Lemma_g_d_tilted_info}.
\end{proof}
Then, we can bound the right-hand side of \eqref{pro_01} from above by
\begin{align}
&\int_{\frac{2 \zeta}{\sqrt{k}}}^{1}\!\! \mathbb{P}\Big[\gamma g_{Z^{\star k}}(X^k, t) < 1\Big] \mathrm{d} t \nonumber \\
\leq& \!\int_{\frac{2 \zeta}{\sqrt{k}}}^{1}\!\! \mathbb{P}\Big[\sum_{i=1}^k \jmath_{X}(X_i,d_s) + k \lambda_s^{\star}\delta_s(t, k) \!+\! C \log k \geq \log \gamma \Big] \mathrm{d} t\nonumber \\
& + \Big(1 - \frac{2 \zeta}{\sqrt{k}}\Big)\frac{T}{\sqrt{k}} \label{by_Fixed_result} \\
\leq&\int_{0}^{1} \mathbb{P}\Big[\sum_{i=1}^k \jmath_{X}(X_i,d_s) + \lambda_s^{\star} \sqrt{k V_1}Q^{-1}(t)\nonumber \\
&+ C\log k \geq \log \gamma \Big] \mathrm{d} t + \Big(1 - \frac{2 \zeta}{\sqrt{k}}\Big)\frac{T}{\sqrt{k}} \label{final_02_new} \\
=&\mathbb{P}\Big[\sum_{i=1}^k \jmath_{X}(X_i,d_s) + \lambda_s^{\star} \sqrt{k V_1}G + C_2\log k \geq \log \gamma \Big] \nonumber \\ 
& + \Big(1 - \frac{2 \zeta}{\sqrt{k}}\Big)\frac{T}{\sqrt{k}} \label{final_03_new} \\
\leq& \epsilon - \frac{2 \zeta + 1}{\sqrt{k}}, \label{final_04_new}
\end{align}
where \eqref{by_Fixed_result} is by Lemma \ref{Lemma_g_d_tilted_info}, \eqref{final_02_new} is by the definition \eqref{delta__s_t_k} and a direct relaxation, in \eqref{final_03_new}, we used \cite{Kostina2016Nonasymptotic}
\begin{align}
\int_{0}^1 1\left\{\mu + v Q^{-1}\left(t\right) > \chi \right\} \mathrm{d}t
= \mathbb{P} \left[\mu + v G > \chi \right] \label{math_insight}
\end{align} 
for scalars $\mu$, $\chi$, random variable $G \sim \mathcal{N}(0,1)$, and $v > 0$, and \eqref{final_04_new} is obtained using \eqref{E_j_equals_R}, \eqref{relationship_V_tildeV}, \eqref{EVAR_e_V1}, the Berry–Ess\'een theorem, and by appropriately choosing the remainder term $O(\log k)$ in \eqref{log_gamma}.

Applying \eqref{log_M}, \eqref{log_gamma}, \eqref{pro_01}, and \eqref{final_04_new} to upper-bound the right-hand side of \eqref{achi_02}, we finally obtain $\epsilon' \leq \epsilon$.

\subsection{Discussion of the Proof}
\label{Discussion}

Our proof is an asymptotic analysis of Lemma \ref{Single_Shot_Achi_one_c}. More specifically, through an asymptotic analysis of the upper bound of the excess-distortion probability provided in Lemma \ref{Single_Shot_Achi_one_c}, we showed that the choice of $M$ in \eqref{log_M}, which satisfies \eqref{second_order_general} in Theorem \ref{theorem_Gaussian_approximation_achi}, is compatible with parameters $k$, $d_s$, and $\epsilon$ in terms of the existence of a $(k, M, d_s, \epsilon)$ code, thereby yielding Theorem \ref{theorem_Gaussian_approximation_achi}. As suggested by the form of the upper bound in Lemma \ref{Single_Shot_Achi_one_c}, our analysis can be naturally divided into three progressive but interrelated parts:
\begin{enumerate}
	\item \label{part_1}
	Addressing the randomness of the hidden source sequence $S^k$ given $x^k$ and $z^k$;
	\item \label{part_2}
	Addressing the randomness of the random codeword $Z^{\star k}$ given $x^k$;
	\item \label{part_3}
	Addressing the randomness of the observable source sequence $X^k$.
\end{enumerate}

In the first part of the proof, we showed that $\bar{\textsf{d}}_s(x^k,z^k) \leq d_s - \delta_s(t, k) $ implies $\pi(x^k,z^k) = \mathbb{P}[\textsf{d}_s(S^k, z^k) > d_s| X^k=x^k] \leq t$ for all sufficiently large $k$ by applying Lemma \ref{tackle_V_lemma} and the Berry–Ess\'een theorem. Note that Lemma \ref{tackle_V_lemma} is a consequence of the fact that $V_k(x^k,z^k)$ is a linear function of $\mu_k(x^k,z^k)$, which relies heavily on the considered squared error distortion and the source model. In fact, if another pair of distortion measure and source model makes $V_k(x^k,z^k)$ a linear function of $\mu_k(x^k,z^k)$, the argument in Lemma \ref{tackle_V_lemma} still holds. Besides, the bounded alphabet assumption guarantees that $T_k(x^k, z^k)$ is uniformly bounded for all $x^k \in \mathcal{X}^k$, $z^k \in \hat{\mathcal{S}}^k$, and $k > 0$, ensuring a constant $\zeta$, which significantly simplifies the analysis.

In the second part of the proof, by utilizing the favorable properties of the generalized rate-distortion optimization problem \eqref{generalized_rate_distortion_function}, we derived Lemma \ref{Lemma_g_d_tilted_info}, thereby establishing a relationship between $g_{Z^{\star k}}(x^k, t)=\mathbb{P}[\bar{\textsf{d}}_s(x^k,Z^{\star k}) \leq d_s - \delta_s(t, k) ]$ and the sum of distortion-tilted information $\sum_{i=1}^k \jmath_{X}(x_i,d_s)$. Lemma \ref{Lemma_g_d_tilted_info} is crucial as it characterizes, within a tolerable error, the perturbation in the required number of bits to encode $x^k$ caused by the remaining randomness of the hidden source given $x^k$. Lemma \ref{Lemma_g_d_tilted_info} also provides a form as a sum of independent random variables, facilitating the application of the Berry-Ess\'een theorem in handling the randomness of $X^k$ in the third part of the proof.

Our proof differs significantly from that in \cite{Kostina2016Nonasymptotic}. In \cite{Kostina2016Nonasymptotic}, under the setting of finite alphabets, the functionality of the first part and part of the second part of our proof can be realized conveniently using the change of measure (in \cite[Theorem 4]{Kostina2016Nonasymptotic}) and the type counting (in \cite[Appendix D]{Kostina2016Nonasymptotic}) arguments. Subsequently, using lemmas analogous to Lemma \ref{Counterpart_02}, they obtain the form of the sum of distortion-tilted information. However, in our continuous alphabet setting, the type-based methods no longer work, leading us to seek new proof techniques. As described above, we avoid applying the non-asymptotic result in \cite[Theorem 4]{Kostina2016Nonasymptotic} and instead address the randomness of $S^k$ using Lemma \ref{tackle_V_lemma}, which holds only asymptotically. Due to this different starting point, as well as the continuous-alphabet nature, our subsequent proof also diverges from that in \cite{Kostina2016Nonasymptotic}. Particularly, in contrast to \cite{Kostina2016Nonasymptotic}, we must provide a more thorough characterization of the impact of $\delta_s(t, k)$ in our subsequent proof, as established in Lemma \ref{Lemma_g_d_tilted_info}.

The proof in this section has been extended to accommodate the simultaneous reconstruction of both sources, as outlined in Appendix \ref{proof_theorem_Gaussian_approximation_achi_two}. The key observation behind this extension is that, given $x^k$, $y^k$, and $z^k$, the newly introduced constraint $\textsf{d}_x(x^k, y^k) \leq d_x$ does not induce any additional randomness. Therefore, the approach used to address the randomness of $S^k$ remains applicable. Subsequently, the randomness of the random codeword $(Z^{\star k}, Y^{\star k})$ given $x^k$ is addressed using Lemma \ref{Lemma_g_d_tilted_info_two}, which is generalized from Lemma \ref{Lemma_g_d_tilted_info}. Finally, the randomness of $X^k$ is still addressed using the Berry–Ess\'een theorem.





\section{Conclusion}
\label{Conclusion}

In this paper, we derived second-order achievability bounds for indirect quadratic lossy source coding of a specific class of continuous-alphabet sources, considering both the cases with and without observed source reconstruction. These bounds indicate that the optimal second-order rates of these problems are bounded above by a term of the order $1/\sqrt{k}$. We noticed that the obtained achievability bounds for these problems share the same forms as their second-order tight finite-alphabet counterparts. Moreover, our results offer a method for approximating the finite blocklength performance of the indirect quadratic lossy source coding problem.

While facilitating the proof, the bounded alphabet assumption makes our results inapplicable to the quadratic Gaussian case, where $(S, X)$ follows a non-degenerate joint Gaussian distribution. Nevertheless, as our setting is only one step away from encompassing the quadratic Gaussian case, and given the favorable properties of the Gaussian distribution, we conjecture that similar second-order results could be obtained for the quadratic Gaussian case as well. This may involve, for instance, developing Gaussian type covering lemmas for the indirect lossy source coding problem. Additionally, there is potential to relax the current constraints on the distortion measure and source model. Exploring how to adapt the proofs under these relaxed conditions presents another promising direction for future research.

\appendices

\section{Auxiliary Results}
\label{Auxiliary_Results}

In this Appendix, we list three lemmas that are crucial to our proof. The form of the Berry–Esséen theorem that we adopt is provided in Lemma \ref{Berry_Esseen}.
\begin{lemma} 
	\label{Berry_Esseen}
	(Berry–Ess\'een CLT, e.g., \cite[Theorem 13]{Kostina2012Fixed}, \cite[Chapter XVI.5, Theorem 2]{Feller1971Introduction}):
	Fix an integer $k > 0$. Let random variables $\{Z_i \in \mathbb{R}\}_{i=1}^k $ be independent. Denote
	\begin{align}
	\mu_k =& \frac{1}{k}\sum_{i=1}^k \mathbb{E}[Z_i],\\
	V_k =& \frac{1}{k}\sum_{i=1}^k \mathrm{Var}[Z_i],\\
	T_k =&\frac{1}{k}\sum_{i=1}^k \mathbb{E}\left[|Z_i - \mathbb{E}[Z_i]|^3\right],\\
	B_k =& \frac{c_0T_k}{V_k^{3/2}}, \label{B_E_B_k}
	\end{align}
	where $c_0$ is a positive constant.
	Then, for any real $t$,
	\begin{align}
	\label{P_Berry_Esseen}
	\left|\mathbb{P}\left[\sum_{i=1}^k Z_i > k \left(\mu_k + t\sqrt{\frac{V_k}{k}}\right)\right] - Q(t)\right| \leq \frac{B_k}{\sqrt{k}},
	\end{align}
	where $Q(t)$ denotes the complementary standard Gaussian cumulative distribution function.
\end{lemma}
Lemmas \ref{Single_Shot_Achi_one_c} and \ref{Single_Shot_Achi_two_c} give the $k$-fold versions of the one-shot achievability bounds in \cite[Theorem 3]{Kostina2016Nonasymptotic} and \cite[Theorem 1]{Yang2024Indirect}, respectively. Theorems \ref{theorem_Gaussian_approximation_achi} and \ref{theorem_Gaussian_approximation_achi_two} are obtained by asymptotic analysis of these bounds. The underlying random-coding-based achievability schemes can be found in the proofs of Lemmas \ref{Single_Shot_Achi_one_c} and \ref{Single_Shot_Achi_two_c}.
\begin{lemma} 
	\label{Single_Shot_Achi_one_c}
	(\cite[Theorem 3]{Kostina2016Nonasymptotic}):
	For any $P_{\bar{Z}^k}$ on $\hat{\mathcal{S}}^k$, there exists a $(k, M, d_s, \epsilon)$ code with
	\begin{equation}
	\label{Single_Shot_Achi_one_c_eq}
	\epsilon \leq \int_{0}^1 \mathbb{E}\Big[\mathbb{P}^M\left[\pi(X^k, \bar{Z}^k) > t | X^k \right]\Big] \mathrm{d} t,
	\end{equation}
	where $P_{X^k \bar{Z}^k} = P_{X^k}P_{\bar{Z}^k}$, and
	\begin{equation}
	\label{pi_func}
	\pi(x^k,z^k) \triangleq \mathbb{P}[\textsf{d}_s(S^k, z^k) > d_s| X^k=x^k].
	\end{equation}
\end{lemma}
\begin{lemma} 
	\label{Single_Shot_Achi_two_c}
	(\cite[Theorem 1]{Yang2024Indirect}):
	For any $P_{\bar{Z}^k\bar{Y}^k}$ on $\widehat{\mathcal{S}}^k \times \widehat{\mathcal{X}}^k$, there exists a $(k, M, d_s,d_x, \epsilon)$ code with
	\begin{equation}
	\label{Single_Shot_Achi_two_c_eq}
	\epsilon \leq \int_{0}^1 \mathbb{E}\Big[\mathbb{P}^M\left[\pi(X^k, \bar{Z}^k, \bar{Y}^k) > t | X^k \right]\Big] \mathrm{d} t,
	\end{equation}
	where $P_{X^k \bar{Z}^k \bar{Y}^k} = P_{X^k} P_{\bar{Z}^k \bar{Y}^k}$ and
	\begin{equation}
	\begin{aligned}
	\pi(x^k,z^k,y^k) =& \mathbb{P}[\textsf{d}_s(S^k, z^k) > d_s\\ 
	&\cup \textsf{d}_x(X^k, y^k) > d_x | X^k=x^k].
	\end{aligned}
	\end{equation}
\end{lemma}

\section{Proof of Lemma \ref{tackle_V_lemma}} \label{proof_tackle_V_lemma}

	In this proof, we abbreviate $\theta(x^k, z^k)$ to $\theta$ for brevity. By the definitions in \eqref{mu_k} and \eqref{V_k}, \eqref{core_condition_02} can be rewritten as
\begin{align}
\label{rewrite_01}
f(\theta| t, k) \triangleq \theta + \sqrt{\frac{4 \sigma_w^2 \theta + \sigma_{w^2}^2}{k}}Q^{-1}\Big(t - \frac{\zeta}{\sqrt{k}}\Big) \leq d_s - \sigma_w^2.
\end{align}

By applying $Q(x) \leq e^{-x^2/2}$, $\forall x > 0$, for all $\frac{2\zeta}{\sqrt{k}} \leq t \leq 1$, we have
\begin{equation}
\label{Q_inv_bound}
\Big|Q^{-1}\Big(t - \frac{\zeta}{\sqrt{k}}\Big)\Big| \leq \sqrt{\ln k - 2 \ln \zeta}.
\end{equation}
By \eqref{Q_inv_bound}, for all $\frac{2\zeta}{\sqrt{k}} \leq t \leq 1$, we have
\begin{align}
&f(d_{s, \min} - \sigma_w^2 | t, k) \nonumber \\
\leq&d_{s, \min} - \sigma_w^2 \nonumber \\ 
&+ \sqrt{\frac{(4 \sigma_w^2 (d_{s, \min} - \sigma_w^2) + \sigma_{w^2}^2) \cdot (\ln k - 2 \ln \zeta)}{k}}. \label{u1_01}
\end{align}
This, combined with the setting $d_s > d_{s, \min}$, implies that there exists $k_1 > 0$ such that for all $k > k_1$,
\begin{align}
\label{condition_1}
f(d_{s, \min} - \sigma_w^2 | t, k) < d_s - \sigma_w^2.
\end{align}

Further, for all $\frac{2\zeta}{\sqrt{k}} \leq t \leq 1$ and $\theta \geq 0$, we have 
\begin{align}
f'(\theta| t, k) =& 1 + \frac{2\sigma_w^2}{\sqrt{(4 \sigma_w^2 \theta + \sigma_{w^2}^2)k}} Q^{-1}\Big(t - \frac{\zeta}{\sqrt{k}}\Big) \label{a_01} \\
\geq& 1 - 2\sigma_w^2\sqrt{\frac{\ln k - 2 \ln \zeta}{(4 \sigma_w^2 \theta + \sigma_{w^2}^2)k}}, \label{a_02}
\end{align}
where \eqref{a_01} follows from a direct computation and \eqref{a_02} is by \eqref{Q_inv_bound}. As a consequence, there exists $k_2 > 0$ such that for all $k > k_2$, 
\begin{align}
\label{condition_2}
f'(\theta| t, k) > 0, \ \forall \theta \geq 0.
\end{align}
Since $f(\theta | t, k)$ is dominated by its first term $\theta$, there exist $\theta_0$, $k_3 > 0$ such that for all $k  > k_3$, we have
\begin{align}
\label{condition_3}
f(\theta_0 | t, k) > d_s - \sigma_w^2.
\end{align}

Note that by the definition of $d_{s, \min}$, there exist $x^k \in \mathcal{X}^k$ and $z^k \in \hat{\mathcal{S}}^k$ such that $\theta(x^k, z^k) + \sigma_w^2 = \frac{1}{k}\sum_{i=1}^k \mathbb{E}[\textsf{d}_s(S_i, z_i)\big|X_i = x_i] \leq d_{s, \min}$. 
Juxtaposing this with \eqref{condition_1}, \eqref{condition_2}, and \eqref{condition_3}, we now can say that, for all $k > \max\{k_1, k_2, k_3 \}$ and $t$ satisfying $\frac{2\zeta}{\sqrt{k}} \leq t \leq 1$, \eqref{rewrite_01} is equivalent to
\begin{equation}
\theta \leq \tilde{\theta},
\end{equation}
where $\tilde{\theta}$ is the unique root of equation
\begin{equation}
\label{f_equation}
f(\theta| t, k) = d_s - \sigma_w^2.
\end{equation}

We now calculate $\tilde{\theta}$. By \eqref{f_equation} and the definition of function $f(\theta| t, k)$, the quadratic equation with respect to $\theta$,
\begin{equation}
(d_s - \sigma_w^2 - \theta)^2 = \frac{4 \sigma_w^2 \theta + \sigma_{w^2}^2}{k}\Big[Q^{-1}\Big(t - \frac{\zeta}{\sqrt{k}}\Big)\Big]^2,
\end{equation}
has two roots, as shown in \eqref{roots} at the top of the next page,
\begin{figure*}[!t]
	\begin{align}
	\label{roots}
	\bar{\theta} = d_s - \sigma_w^2 + \frac{2\sigma_w^2 \Big[Q^{-1}\Big(t - \frac{\zeta}{\sqrt{k}}\Big)\Big]^2}{k}
	\pm \sqrt{\frac{V_1}{k} + \frac{4 \sigma_w^4 \Big[Q^{-1}\Big(t - \frac{\zeta}{\sqrt{k}}\Big)\Big]^2}{k^2}} \Big|Q^{-1}\Big(t - \frac{\zeta}{\sqrt{k}}\Big)\Big|
	\end{align}
	\hrulefill
\end{figure*}
and $\tilde{\theta}$ is one of them. Since the two roots are the same when $t = 1/2 + \frac{\zeta}{\sqrt{k}}$, we now focus on the case when $t \neq 1/2 + \frac{\zeta}{\sqrt{k}}$. By \eqref{Q_inv_bound}, there exists $k_4 > 0$ such that for $k > k_4$ and $t \neq 1/2 + \frac{\zeta}{\sqrt{k}}$, $(d_s - \sigma_w^2 - \bar{\theta})$ is always opposite in sign to the last term on the right side of \eqref{roots}. Combining this with the fact that $(d_s - \sigma_w^2 - \tilde{\theta})$ is always of the same sign as $Q^{-1}\big(t - \frac{\zeta}{\sqrt{k}}\big)$ (by \eqref{f_equation}), we obtain \eqref{mu_u_02}.
\begin{figure*}[!t]
	\begin{align}
	\label{mu_u_02}
	\tilde{\theta}
	= d_s - \sigma_w^2 + \frac{2\sigma_w^2 \Big[Q^{-1}\Big(t - \frac{\zeta}{\sqrt{k}}\Big)\Big]^2}{k}
	- \sqrt{\frac{V_1}{k} + \frac{4 \sigma_w^4 \Big[Q^{-1}\Big(t - \frac{\zeta}{\sqrt{k}}\Big)\Big]^2}{k^2}} Q^{-1}\Big(t - \frac{\zeta}{\sqrt{k}}\Big)
	\end{align}
	\hrulefill
\end{figure*}

It remains to show that there exists $k_5>0$ such that $d_s - \sigma_w^2 - \delta_s(t, k) \leq \tilde{\theta}$ for all $k > k_5$. This is true by noting that $\tilde{\theta}$ can be written as 
\begin{align}
	\tilde{\theta} = d_s - \sigma_w^2 - \delta_s(t, k)  + O\Big(\frac{\log k}{k}\Big),
\end{align}
where the term $O\big(\frac{\log k}{k}\big)$ is dominated by the term $\frac{2\sigma_w^2 [Q^{-1}(t - \frac{\zeta}{\sqrt{k}})]^2}{k}$ and thus is positive for all $k > k_5$ with a properly chosen $k_5$. The proof is finally completed by letting $k_0 = \max\{k_1,k_2,k_3,k_4,k_5\}$.

\section{Proof of Lemma \ref{Lemma_g_d_tilted_info}} \label{proof_Lemma_g_d_tilted_info}

This proof is generalized from the proof of \cite[Lemma 2]{Kostina2012Fixed} and we provide only a brief sketch. We begin by introducing several necessary definitions. Define optimization problem
\begin{subequations}
	\label{generalized_rate_distortion_function}
	\begin{align}
	R_{X,Y}(d) \triangleq \min_{P_{Z|X}}\ &D(P_{Z|X}||P_Y|P_X) \\
	\mathrm{s.t.}\
	&\mathbb{E}\left[\bar{\textsf{d}}_s(X, Z)\right] \leq d.
	\end{align}
\end{subequations}
Define for any $\lambda \geq 0$
\begin{align}
\label{Lambda_Y_d}
\Lambda_{Y, d}(x,\lambda) \triangleq \log \frac{1}{\mathbb{E}[\exp(\lambda d - \lambda d(x,Y) )]}.
\end{align}
For $d > d_{\min|X,Y} \triangleq \inf\{d: R_{X,Y}(d) < \infty \}$, $R_{X,Y}(d)$ is always achieved by a $P_{Z^\star|X}$ that satisfies \cite{Kostina2012Fixed}
\begin{align}
\label{generalized_tilted_information_property}
\log \frac{\mathrm{d} P_{Z^\star| X=x}}{\mathrm{d} P_{Y}}(y) =& \Lambda_{Y, d}(x,\lambda_{X,Y,d}^{\star}) \nonumber \\ 
&- \lambda_{X,Y,d}^{\star} d(x,y) + \lambda_{X,Y,d}^{\star} d,
\end{align}
where $\lambda_{X,Y,d}^{\star} = - R_{X,Y}^{\prime}(d)$.

Given the above notations, following the proof of \cite[Lemma 1]{Kostina2012Fixed}, for any $\nu > 0$, $P_{\hat{X}}$ on $\mathcal{X}$, and $P_{Y}$ on $\hat{\mathcal{S}}$ such that $d_{\min|\hat{X},Y} \leq d_s - \delta_s(t, k)$, we have
\begin{align}
&\mathbb{P}[\bar{\textsf{d}}_s(x, Y) \leq d_s - \delta_s(t, k)]  \nonumber \\
\geq& \exp\big(- \Lambda_{Y, d_s - \delta_s(t, k) }\big(x,\lambda_{\hat{X},Y,d_s - \delta_s(t, k)}^{\star}\big) - \lambda_{\hat{X},Y,d_s - \delta_s(t, k)}^{\star} \nu \big) \nonumber \\ 
&\cdot \mathbb{P}\big[d_s - \delta_s(t, k) - \nu < \bar{\textsf{d}}_s(x, Z^{\star}) \leq d_s - \delta_s(t, k) | \hat{X} = x \big], \label{Lemma4_01} 
\end{align}
where $P_{Z^\star|\hat{X}}$ achieves $R_{\hat{X},Y}(d_s - \delta_s(t, k))$.

By setting $\nu = \tau / k$ with fixed $\tau > 0$, $P_{Y} = P_{Z^{\star k}}$, $P_{\hat{X}} = P_{\hat{X}^k} = P_{\hat{X}} \times \dots \times P_{\hat{X}}$, where $P_{\hat{X}}$ is the measure on $\mathcal{X}$ generated by the empirical distribution of $x^k \in \mathcal{X}^k$, i.e., 
\begin{equation}
\label{P_hat_X}
P_{\hat{X}}(a) = \frac{1}{k} \sum_{i=1}^k 1\{x_i = a\},
\end{equation}
we have \eqref{Lemma4_02} and \eqref{Lemma4_03} at the top of the next page,
\begin{figure*}[!t]
	\begin{align}
	&\mathbb{P}\Big[\frac{1}{k}\sum_{i=1}^k\bar{\textsf{d}}_s(x_i, Z_i^{\star}) \leq d_s - \delta_s(t, k)\Big]  \nonumber \\
	\geq&\exp\Big(- \sum_{i=1}^k \Lambda_{Z^{\star}, d_s - \delta_s(t, k) }\big(x_i,\lambda_{\hat{X},Z^{\star},d_s - \delta_s(t, k)}^{\star}\big)
	- \lambda_{\hat{X},Z^{\star},d_s - \delta_s(t, k)}^{\star} \tau \Big) \cdot \mathbb{P}\Big[kd_s - k\delta_s(t, k) - \tau \nonumber \\ 
	&< \sum_{i=1}^k\bar{\textsf{d}}_s(x_i, Z_i^{\star}) \leq kd_s - k\delta_s(t, k) | \hat{X}^k = x^k \Big], \label{Lemma4_02}            \\
	=&\exp\Big(- \sum_{i=1}^k \Lambda_{Z^{\star}, d_s }\big(x_i,\lambda_{\hat{X},Z^{\star},d_s - \delta_s(t, k)}^{\star}\big) 
	- k \lambda_{\hat{X},Z^{\star},d_s - \delta_s(t, k)}^{\star} \delta_s(t, k) - \lambda_{\hat{X},Z^{\star},d_s - \delta_s(t, k)}^{\star} \tau \Big) \nonumber \\  
	&\cdot \mathbb{P}\Big[kd_s - k\delta_s(t, k) - \tau 
	< \sum_{i=1}^k\bar{\textsf{d}}_s(x_i, Z_i^{\star}) 
	\leq kd_s - k\delta_s(t, k) | \hat{X}^k = x^k \Big], \label{Lemma4_03}
	\end{align}
	\hrulefill
\end{figure*}
where $P_{Z^{\star k}|\hat{X}^k} = P^{\otimes k}_{Z^{\star}|\hat{X}}$, \eqref{Lemma4_02} is by \eqref{Lemma4_01}, the settings, and the fact that $\Lambda_{Z^{\star}, d}(x^k,k \lambda) = \sum_{i=1}^k \Lambda_{Z^{\star}, d}(x_i,\lambda)$ for any $\lambda >0$, and \eqref{Lemma4_03} is by the definition \eqref{Lambda_Y_d}.

We now give the following two lemmas to tackle the right-hand side of \eqref{Lemma4_03}. Note that these lemmas are obtained by slightly modifying the proofs of \cite[Lemmas 4, 5]{Kostina2012Fixed}.
\begin{lemma} 
	\label{Counterpart_01}
	Fix $d_{s, \min} < d_s < d_{s,\max}$. There exist $\kappa_0$, $k_7 > 0$ such that for all $\kappa \leq \kappa_0$, $k \geq k_7$, $\frac{2\zeta}{\sqrt{k}} \leq t \leq 1$, there exist a set $F_k \subseteq \mathcal{X}^k$ and constants $\tau$, $C_1$, $K_1 > 0$ such that
	\begin{equation}
	\mathbb{P}[X^k \notin F_k] \leq \frac{K_1}{\sqrt{k}}
	\end{equation}
	and for all $x^k \in F_k$
	\begin{align}
	\mathbb{P}\Big[&kd_s - k\delta_s(t, k) - \tau 
	< \sum_{i=1}^k\bar{\textsf{d}}_s(x_i, Z_i^{\star}) \nonumber \\   
	&\leq kd_s - k\delta_s(t, k) | \hat{X}^k = x^k \Big] \geq \frac{C_1}{\sqrt{k}},
	\end{align}
	\begin{align}
	|\lambda_{\hat{X},Z^{\star},d_s - \delta_s(t, k)}^{\star} - \lambda_s^{\star}| < \kappa,
	\end{align}
	where $\lambda_{\hat{X},Z^{\star},d_s - \delta_s(t, k)}^{\star}$ depends on $x^k$ through \eqref{P_hat_X}.
\end{lemma}
\begin{proof}[Proof]
	By \eqref{Q_inv_bound}, we have for all $\frac{2\zeta}{\sqrt{k}} \leq t \leq 1$,
	\begin{equation}
	\label{order_of_delta_s_new}
	|\delta_s(t, k)| \leq \sqrt{\frac{V_1(\ln k - 2 \ln \zeta)}{k}}.
	\end{equation} 
	Juxtaposing \eqref{order_of_delta_s_new} with the fact that $R_{\hat{X},Z^{\star}}(d)$ is continuously differentiable at point $d=d_s$ \cite{En1999On}, we have $\lambda_{\hat{X},Z^{\star},d_s - \delta_s(t, k)}^{\star}$ can arbitrarily close to $\lambda_{\hat{X},Z^{\star},d_s}^{\star}$ as $k \to \infty$. Therefore, $\lambda_{\hat{X},Z^{\star},d_s - \delta_s(t, k)}^{\star}$ satisfies \cite[(295)]{Kostina2012Fixed} for all sufficiently large $k$ if $\lambda_{\hat{X},Z^{\star},d_s}^{\star}$ satisfies \cite[(295)]{Kostina2012Fixed}. Keeping this in mind and proceeding in the same manner with the proof of \cite[Lemma 4]{Kostina2012Fixed}, Lemma \ref{Counterpart_01} follows.
\end{proof}

\begin{lemma} 
	\label{Counterpart_02}
	Fix $d_{s, \min} < d_s < d_{s,\max}$. There exist constants $k_8$, $K_2$, $C_2 > 0$ such that for all $k \geq k_8$ and $\frac{2\zeta}{\sqrt{k}} \leq t \leq 1$, we have
	\begin{align}
	&\mathbb{P}\Bigg[\sum_{i=1}^k \! \Big(\Lambda_{Z^{\star}, d_s }\big(X_i,\lambda_{\hat{X},Z^{\star},d_s - \delta_s(t, k)}^{\star}\big) + \lambda_{\hat{X},Z^{\star},d_s - \delta_s(t, k)}^{\star} r_k \Big)\nonumber \\ 
	&\ \ \ \leq \sum_{i=1}^k \Big(\Lambda_{Z^{\star}, d_s }\big(X_i,\lambda_{s}^{\star}\big) + \lambda_{s}^{\star} r_k \Big) + C_2\log k \Bigg] > 1 - \frac{K_2}{\sqrt{k}}, \label{Lemma_Lambda}
	\end{align}
	where $r_k$ satisfies $r_k = O\big(\sqrt{\log k/k}\big)$.
\end{lemma}
\begin{proof}[Proof]
	Given Lemma \ref{Counterpart_01}, following the proof of \cite[Lemma 5]{Kostina2012Fixed}, \eqref{Lemma_Lambda} with $r_k = 0$ is readily obtained. With slight adaptations to the proof of \cite[Lemma 5]{Kostina2012Fixed}, it can further be shown that \eqref{Lemma_Lambda} holds for $r_k = O\big(\sqrt{\log k/k}\big)$ as well.
\end{proof}

Recall that $\delta_s(t, k) = O\big(\sqrt{\log k/k}\big)$ for $\frac{2\zeta}{\sqrt{k}} \leq t \leq 1$, and thus it satisfies the restriction on $r_k$ in Lemma \ref{Counterpart_02}. Note that $\Lambda_{Z^{\star}, d_s }\big(x,\lambda_{s}^{\star}\big) = \jmath_{X}(x,d_s)$. Applying Lemmas \ref{Counterpart_01} and \ref{Counterpart_02} to relax \eqref{Lemma4_03}, we finally arrive at Lemma \ref{Lemma_g_d_tilted_info}.

\section{Proof of Theorem \ref{theorem_Gaussian_approximation_achi_two}} \label{proof_theorem_Gaussian_approximation_achi_two}

This proof is similar to that of Theorem \ref{theorem_Gaussian_approximation_achi} and we only give a proof sketch. Unless otherwise specified, the notations here are reused from Section \ref{proof_theorem_Gaussian_approximation_achi} with the same meanings. The proof is an asymptotic analysis of Lemma \ref{Single_Shot_Achi_two_c} in Appendix \ref{Auxiliary_Results}. By Lemma \ref{Single_Shot_Achi_two_c} and the arguments leading to \eqref{achi_01} and \eqref{achi_02}, for any $\gamma > 0$, there exists a $(k, M, d_s, d_x, \epsilon')$ code such that
\begin{align}
\epsilon^{\prime} \leq& e^{-\frac{M}{\gamma}} + \int_{\frac{2 \zeta}{\sqrt{k}}}^{1} \mathbb{P}\Big[\gamma \mathbb{P}\left[\pi(X^k, Z^{\star k}, Y^{\star k}) \leq t | X^k \right] < 1\Big] \mathrm{d} t \nonumber \\ 
&+\frac{2 \zeta}{\sqrt{k}}, \label{achi_two_02}
\end{align}
where $P_{X^k Z^{\star k} Y^{\star k}} = P_{X^k} P_{Z^{\star k} Y^{\star k}}$, $P_{X^k} = P_X \times \dots \times P_X$, $P_{Z^{\star k}Y^{\star k}} = P_{Z^{\star}Y^{\star}} \times \dots \times P_{Z^{\star}Y^{\star}}$ with $P_{Z^{\star}Y^{\star}}$ achieving  $R_{S,X}(d_s, d_x)$. Similarly, let
\begin{align}
\log M =& \log \gamma + \log \ln \sqrt{k}, \label{log_M_two} \\
\log \gamma =& kR_{S,X}(d_s,d_x) + \sqrt{k \tilde{\mathcal{V}}(d_s,d_x)} Q^{-1}(\epsilon) \nonumber \\  
&+ O(\log k). \label{log_gamma_two}
\end{align}
Observe that $\log M$ takes the desired form, i.e., the right-hand side of \eqref{second_order_general_two}. In the following, we aim to show that the right-hand side of \eqref{achi_two_02}, with a properly chosen remainder term (i.e., the term $O(\log k)$ in \eqref{log_gamma_two}), can be upper bounded by $\epsilon$. This implies that the choice of $M$ in \eqref{log_M_two} is compatible with parameters $k$, $d_s$, $d_x$, and $\epsilon$, thereby proving the existence of a $(k, M, d_s, d_x, \epsilon)$ code.


For $(x^k,z^k, y^k)$ and $t$ such that
\begin{align}
&\textsf{d}_x(x^k, y^k) \leq d_x, \label{d_x_condition} \\
&\mu_k(x^k,z^k) \leq d_s - \sqrt{\frac{V_k(x^k,z^k)}{k}}Q^{-1}\Big(t - \frac{\zeta}{\sqrt{k}}\Big), \label{core_condition_two_02} \\
& \frac{\zeta}{\sqrt{k}} < t < 1, \label{core_condition_two_02_01}
\end{align}
we have
\begin{align}
\pi(x^k,z^k, y^k) = \pi(x^k,z^k) \leq  t, \label{By_berry_esseen_two}
\end{align}
where the equality holds due to condition \eqref{d_x_condition}, while the inequality is derived from the reasoning in \eqref{By_Condition_1}-\eqref{By_berry_esseen}. Note that for all $d_s$ such that $(d_s, d_x) \in \textsf{int}(\mathcal{D}_{\mathrm{adm}})$, we have $d_s > d_{s, \min}$. Consequently, we can still tackle \eqref{core_condition_two_02} using Lemma \ref{tackle_V_lemma}. Therefore, defining
\begin{align}
&g_{Z^{\star k} Y^{\star k}}(x^k, t) \triangleq \mathbb{P}\bigg[\textsf{d}_x(x^k, Y_i^{\star k}) \leq d_x \nonumber \\ 
&\cap \frac{1}{k}\sum_{i=1}^k(\varphi(x_i) - Z^{\star}_i)^2 + \sigma_w^2 \leq d_s - \delta_s(t, k) \bigg],
\end{align}
we have by \eqref{d_x_condition}-\eqref{By_berry_esseen_two} and Lemma \ref{tackle_V_lemma},
\begin{align}
\label{core_01_two}
\mathbb{P}[\pi(x^k, Z^{\star k}, Y^{\star k}) \leq t] 
\geq g_{Z^{\star k} Y^{\star k}}(x^k, t)
\end{align}
for all $k > k_0$ and $\frac{2\zeta}{\sqrt{k}} \leq t \leq 1$.

The counterpart of Lemma \ref{Lemma_g_d_tilted_info} can be stated as follows. The proof of this lemma is obtained by extending the proof of Lemma \ref{Lemma_g_d_tilted_info} using \cite[Lemmas 3, 5]{Yang2024Indirect} and is omitted here.
\begin{lemma} 
	\label{Lemma_g_d_tilted_info_two}
	There exist constants $C$, $T$, $k_6 > 0$ such that for all $k > k_6$, $(d_s, d_x) \in \textsf{int}(\mathcal{D}_{nd})$, and $\frac{2\zeta}{\sqrt{k}} \leq t \leq 1$,
	\begin{align}
	\label{Fixed_result_two}
	\mathbb{P}\bigg[&\log\frac{1}{g_{Z^{\star k} Y^{\star k}}(X^k, t)} \leq \sum_{i=1}^k \jmath_{X}(X_i,d_s, d_x) + k \lambda_s^{\star}\delta_s(t, k) \nonumber \\  
	&+ C \log k \bigg] 
	\geq 1 - \frac{T}{\sqrt{k}},
	\end{align}
	where $\lambda_s^{\star}$ is that defined in \eqref{lambda_s_two}.
\end{lemma}

Further, by the definition \eqref{Definition_of_Dnd} of $\mathcal{D}_{nd}$, we have $\mathbb{E}\left[\bar{\textsf{d}}_s(X,Z^\star)\right] = d_s$, which implies that $\mathbb{E}\left[\textrm{Var}\left[\textsf{d}_s(S,Z^\star)|X,Z^\star\right]\right] = V_1$ still holds for $(d_s,d_x) \in \mathcal{D}_{nd}$. Up to this point, all conditions for relaxing the middle term of the right-hand side of \eqref{achi_two_02} have been satisfied. Similar to \eqref{pro_01} and \eqref{by_Fixed_result}-\eqref{final_04_new}, we have
\begin{align}
\int_{\frac{2 \zeta}{\sqrt{k}}}^{1} \mathbb{P}\Big[\gamma \mathbb{P}\left[\pi(X^k, Z^{\star k}, Y^{\star k}) \leq t | X^k \right] < 1\Big] \mathrm{d} t
\leq \epsilon - \frac{2 \zeta + 1}{\sqrt{k}}. \label{final_04_new_two}
\end{align}
Juxtaposing \eqref{achi_two_02}, \eqref{log_M_two}, \eqref{log_gamma_two}, and \eqref{final_04_new_two}, we finally obtain $\epsilon^{\prime} \leq \epsilon $, which completes the proof.

\bibliographystyle{IEEEtran}
\bibliography{ref}

\end{document}